\documentclass[11pt]{article}
\usepackage{amsfonts}
\usepackage{latexsym,amssymb,amsmath,graphics,cite,arydshln}
\usepackage{tikz}
\usepackage{fancyhdr}
\usepackage{lipsum}
\usepackage{lineno}
\usepackage{color}
\usepackage{diagbox}
\usepackage{colortbl}
\usepackage{multirow}

\begin{document}
\newcommand{\qed}{\hphantom{.}\hfill $\Box$\medbreak}
\newcommand{\proof}{\noindent{\bf Proof \ }}

\newtheorem{theorem}{Theorem}[section]
\newtheorem{lemma}[theorem]{Lemma}
\newtheorem{corollary}[theorem]{Corollary}
\newtheorem{remark}[theorem]{Remark}
\newtheorem{example}[theorem]{Example}
\newtheorem{definition}[theorem]{Definition}
\newtheorem{construction}[theorem]{Construction}
\newtheorem{fact}[theorem]{Fact}
\newtheorem{proposition}[theorem]{Proposition}
\newtheorem{conjecture}[theorem]{Conjecture}
\newtheorem{claim}[theorem]{Claim}

\newenvironment{poc}{\begin{proof}[Proof of the claim]\renewcommand*{\qedsymbol}{$\blacksquare$}}{\end{proof}}
\begin{center}
{\Large\bf Bounds and constructions for constant/low-power error -correcting cooling codes}
\renewcommand{\thefootnote}{}
\footnotetext{This work of Shuangqing Liu is supported by The Natural Science Foundation of the Jiangsu Higher Education Institutions of China under Grant 25KJB110016. 
This work of Tingting Tong is supported by the National Key Research and Development Program of China under Grant 2021YFA1001000, the National Natural Science Foundation of China under Grants 12571576 and 12231014, and the Shandong Provincial Natural Science Foundation under Grants ZR2025QA05 and ZR2026ZD36.
This work of Binwei Zhao is supported by the National Natural Science Foundation of China under Grant 12271023.
(Corresponding author: Tingting Tong)}

\vskip12pt
Shuangqing Liu$^1$,  Tingting Tong$^2$,
 Menglong Zhang$^3$, Binwei Zhao$^4$, Yuhao Zhao$^5$ \\[2ex] {\footnotesize $^1$Department of Mathematics, Suzhou University of Science and Technology, Suzhou 215009, P. R. China}\\
  {\footnotesize$^2$ The State Key Laboratory of Cryptography and Digital Economy Security, the Key Laboratory of Cryptologic Technology and Information Security, Ministry of Education, and the School of Cyber Science and Technology, Shandong University, Qingdao, Shandong 266237, P. R. China. }\\
 {\footnotesize$^3$ Institute of Mathematics and Interdisciplinary Sciences, Xidian University, Xi'an 710126, P.R. China}\\
 {\footnotesize$^4$ School of Mathematics and Statistics, Beijing Jiaotong University, Beijing 100044, P. R. China}\\
 {\footnotesize$^5$ School of Mathematical Sciences, University of Science and Technology of China, Hefei 230026, P. R. China}\\

{\footnotesize
shuangqingliu@usts.edu.cn, tingtingtong2000@163.com, zhangmenglong@xidian.edu.cn, binweizhao@qq.com, zhaoyh21@mail.ustc.edu.cn}
\vskip12pt
\end{center}

\vskip12pt

\noindent {\bf Abstract:}
The low-power error-correcting cooling (LPECC) codes and constant-power error-correcting cooling (CPECC) codes, introduced in [IEEE  Trans. Inf. Theory, 64 (2018), 3062--3085; 66 (2020), 4804--4818], respectively, are two coding schemes designed to simultaneously control the peak temperature and average power consumption of on-chip buses while providing error-correction capability for transmitted information.
This paper establishes new upper bounds for both $(n,1,w,w-2)$-CPECC codes and $(n,t,w,w-2)$-CPECC codes using graph-theoretic techniques, and constructs several new families of optimal CPECC codes using combinatorial configurations. Moreover, it completely resolves the conjecture concerning CPECC codes posed in [IEEE  Trans. Inf. Theory, DOI: 10.1109/TIT.2026.3721101].  Finally, we derive a new upper bound for $(n,t,w,w-2)$-LPECC codes by probabilistic method, along with new optimal families, and establish the relationship between optimal $(n,t,w,w-2)$-LPECC codes and optimal $(n+1,t,w,w-2)$-CPECC codes.



\section{Introduction}
Power and heat dissipation have become major design constraints for modern chips, including both battery-powered devices and high-performance systems. Various encoding techniques have been proposed in the literature \cite{bmmss,ccl,po,sm} to reduce the overall power consumption of on-chip and off-chip buses. However, such techniques do not directly address peak temperature minimization. Hence, power-aware design alone is generally insufficient for effective thermal management, as it does not explicitly control the spatial and temporal distribution of heat.

Chee et al. \cite{1} introduced a new class of codes, called cooling codes, to directly control the peak temperature of a bus by cooling its hottest wires. This is achieved by avoiding state transitions on the hottest wires for as long as necessary until their temperature drops off. Cooling codes are based on differential encoding. 
Specifically, suppose that the current state of the wires is
\((s_1,s_2,\ldots,s_n)\), where wire \(i\) is in state \(s_i\), and that the binary vector \((x_1,x_2,\ldots,x_n)\) is to be transmitted. The wires are then updated to the state \((s'_1,s'_2,\ldots,s'_n)\), where \(s'_i=s_i+x_i \pmod 2,~ 1\le i\le n.\)
Thus, a state transition occurs on wire \(i\) if and only if \(x_i=1\).

To simultaneously control the peak temperature and the average power
consumption, Chee et al.~\cite{1} further introduced low-power cooling (LPC) codes. 
In addition, if an LPC code can control  error-correction for the transmitted information, then it called a low-power error-correcting cooling (LPECC) code.
Subsequently, Chee et al.~\cite{cekvw} introduced constant-power error-correcting cooling (CPECC) codes, in which the number of state transitions is required to remain constant in each transmission.

In this paper, we are particularly interested in LPECC codes and CPECC codes.
An $(n,t,w,e)$-LPECC code \cite{1} is a coding scheme for communication over a bus consisting of $n$ wires satisfying the following three properties:
\begin{itemize}
\item[\textbf{A}($t$):] every transmission does not cause state transitions on the $t$ hottest wires;
\item[\textbf{B}($w$):] the total number of state transitions on all the wires is at most $w$, in every transmission;
\item[\textbf{C}($e$):] up to $e$ transmission errors (0 received as 1, or 1 received as 0) on the $n$ wires can be corrected.
\end{itemize}
If Property $\textbf{B}(w)$ is replaced with
$\textbf{B}'(w)$: the total number of state transitions on all the wires is exactly $w$, in every transmission,
then the LPECC code is called an $(n,t,w,e)$-CPECC code. Note that if the coding scheme only satisfies Property $\textbf{A}(t)$, the code is an $(n, t)$-cooling code; if the coding scheme satisfies Property $\textbf{A}(t)$ and simultaneously Property $\textbf{B}(w)$, the code is an $(n,t,w)$-low-power cooling (LPC) code.

Chee et al. \cite{1} introduced three methods to construct LPECC codes based on complete hypergraphs, dual codes and sunflowers. Liu and Ji~\cite{lj} presented some upper bounds and several families of optimal LPECC codes using combinatorial configurations. Chee et al. \cite{cekvw} also constructed CPECC codes by some special linear codes, for example Reed-Solomon code. Zhao and Zhang \cite{z-z} also gave some asymptotic results of LPECC codes and CPECC codes. More recently, Liu et al. \cite{lyj} showed new optimal $(n,t,4,2)$-CPECC codes and a conjecture on CPECC codes.
 Motivated by these results, we further investigate upper bounds and optimal constructions for LPECC and CPECC codes, and resolve the conjecture posed in \cite{lyj}. The known and new results are summarized in Table~I.

\begin{table}[!htbp]
\begin{center}
{\small \centerline{\footnotesize TABLE I:  KNOWN AND NEW UPPER/LOWER BOUNDS OF LPECC/CPECC CODES}

\vspace{0.3cm}
\begin{tabular}{|c|c|c|}
\hline $C(n,t,w,e)$/$C'(n,t,w,e)$ & Constraints & Reference \\ \hline
\multirow{2}{*}{$C(n,t,w,w-2)> \frac{n-1}{w-1}$} & $n=w(t+1)\equiv w\pmod {w(w-1)}$,  & \multirow{2}{*}{\cite{1}} \\
& and $n$ is sufficiently large& \\ \hline
$C(n,t,4,2)>\frac{n-1}{3}$ & $n=4(t+1)\equiv4\pmod {12}$ & \cite{1} \\ \hline
$ C(n,t,4,2)\geq\frac{n^2 }{12(t+1)}(1-o(1))$ & None & \cite{z-z} \\ \hline
$C(n,t,3,1)=\frac{n^2}{6(t+1)}(1-o(1))$ & None & \cite{z-z} \\ \hline
\multirow{2}{*}{$C(n,t,w,w-2)=\frac{{n+1\choose 2}}{{w+t\choose 2}}$} & $w+t-1|n$, $w\geq t^2+2t+2\geq 5$,  & \multirow{2}{*}{\cite{z-z}} \\
& ${w+t\choose 2}|{n+1\choose 2}$ and $n$ is sufficiently large& \\ \hline
\multirow{3}{*}{$C(n,t,4,2)= \big\lfloor\frac{n(n+1)}{10(t+1)}\big\rfloor$} & $t\equiv 1 \pmod {2}$, $n\equiv 4 \pmod {20}$ & \multirow{3}{*}{\cite{lyj}} \\
& $n\equiv 4\pmod {\frac{5(t+1)}{2}}$, $n\geq 5(t+1)$ & \\
& and $n\notin \{44, 344, 464, 644\}$ & \\ \hline
$C(n,1,3,1)=\big\lfloor\frac{n(n+1)}{12}\big\rfloor$ & $n\neq6$ & \cite{lj} \\ \hline
$C(n,2,3,1)=\big\lfloor\frac{n(n+1)}{18}\big\rfloor$ & $n\neq 6,7, 9$ & {\cite{lyj,lj}}\\ \hline
 $C(n,t,w,w-2)\le \frac{n(n+1)}{2(2w-3)(t+1)}$ & $n\ge (2w-3)(t+1)$ & Theorem \ref{thm:lpecc-upper}\\ \hline
  \multirow{5}{*}{$C(n,t,w,w-2)= \frac{n(n+1)}{2(2w-3)(t+1)}$} &$n> e^{e^{(2w-3)^{12(2w-3)^2}}}-1$,  & \multirow{5}{*}{Corollary \ref{cor-4}} \\
& $n\ge (2w-3)(t+1)-1$,& \\
& $t=l(w-2)-1$,& \\
& $n\equiv 0\pmod {(2w-4)}$, & \\
& $n\equiv -1\pmod {l(2w-3)}$& \\ \hline
\multirow{2}{*}{$C(n,t,w,w-2)= \left\lfloor\frac{n(n+1)}{2(2w-3)(t+1)}\right\rfloor $} & $n$ is sufficiently large,
& \multirow{2}{*}{  Corollary \ref{cor-5}} \\
& $(i$) or $(ii)$ or $(iii)$ in Corollary \ref{cor-5}& \\  \hline
$C'(qw,q-1,w,e)\geq q^{w-e-1}$ & $q$ is a prime power & \cite{cekvw} \\ \hline
\multirow{3}{*}{ $ C'(n,t,w,e)= \frac{\binom{n}{w-e}}{\binom{w+t}{w-e}}$} & $n$ is sufficiently large, $e\leq w-1$, & \multirow{3}{*}{\cite{z-z}} \\
&  $2\binom{w}{w-e}\geq \binom{w-t-1}{w-e}+\binom{w+t}{w-e}$,  &  \\
&  $\binom{w+t-i}{w-e-i}|\binom{n-i}{w-e-i}$, $0\leq i\leq w-e-1$,  &  \\  \hline
\multirow{2}{*}{$C'(n,t,4,2)= \frac{n(n-1)}{10(t+1)}$}& $t=2\ell-1$, $5\ell|n$, $n\equiv  5 \pmod {20}$,
& \multirow{2}{*}{ {\cite{lyj}}} \\
& $n\neq45, 345, 465, 645$ & \\  \hline
{$C'(n,1,w,w-2)\leq h_1 $}& $w\geq 4$ & Theorem \ref{thm:ub_C'(n,1,w,w-2)}
 \\  \hline
 \multirow{3}{*}{$C'(n,1,w,w-2)= h_2$}& $n\geq e^{{(w+1)}^{{(w+1)}^{2}}}-1$,
& \multirow{3}{*}{ Corollary \ref{cor:C'(n,1,w,w-2)}} \\
& $w\geq4$, $n\equiv 0\pmod {w}$,& \\
& $n(n+1)\equiv 0\pmod {w(w+1)}$& \\  \hline
{$C'(n,t,w,w-2)\leq  \frac{n(n-1)}{2(2w-3)(t+1)} $}& None & Theorem \ref{thm:ub_(n,t,w,w-2)-CPECC}
 \\  \hline
 \multirow{4}{*}{$C'(n,t,w,w-2)= \frac{n(n-1)}{2(2w-3)(t+1)}$} &$n> e^{e^{(2w-3)^{12(2w-3)^2}}}$,   &\multirow{4}{*}{Corollary \ref{cor:coj_solve_2}} \\
& $t=l(w-2)-1$,& \\
& $n\equiv 1\pmod {(2w-4)}$, & \\
& $n\equiv 0\pmod {l(2w-3)}$ & \\ \hline
\multirow{2}{*}{$C'(n,t,w,w-2)= \left\lfloor\frac{n(n-1)}{2(2w-3)(t+1)}\right\rfloor $} & $n$ is sufficiently large,
& \multirow{2}{*}{ Theorem \ref{thm:CPECC_construct_3}} \\
& $(i$) or $(ii)$ or $(iii)$ in Theorem \ref{thm:CPECC_construct_3}& \\  \hline
\end{tabular}}\end{center}
\smallskip
where $h_1=\left\lfloor\frac{n}{w+1}\left\lfloor\frac{n-1}{w}\right\rfloor\right\rfloor+\Big\lfloor\frac{n(n-1)-
	\left\lfloor\frac{n}{w+1}\left\lfloor\frac{n-1}{w}\right\rfloor\right\rfloor(w+1)w}{2w(w-1)}\Big\rfloor$, $h_2=\left\lfloor\frac{n}{w+1}\left\lfloor\frac{n-1}{w}\right\rfloor\right\rfloor+\left\lfloor\frac{n}{2w}\right\rfloor $.
\end{table}

The rest of this paper is organized as follows.
In Section 2, we introduce the LPECC codes, CPECC codes, and some related combinatorial configurations.
In Section 3, we firstly show a tighter upper bound of $(n,1,w,w-2)$-CPECC codes by the clique number of graph (see Lemmas \ref{lem:key-lemma}, \ref{lem:stability}, and Theorem \ref{thm:ub_C'(n,1,w,w-2)}), and their optimal families by BIBDs (see Theorem \ref{thm:bibd_C'(n,1,w,w-2)} and Corollary \ref{cor:C'(n,1,w,w-2)}).
We also present a new upper bound of $(n,t,w,w-2)$-CPECC codes by the $K_w$-transversal number and induced degenerate subgraphs of graph (see Lemmas \ref{lem:K_w-free}, \ref{lem:CPECC_e(G_i)}, and Theorem \ref{thm:ub_(n,t,w,w-2)-CPECC}), and also construct some families of optimal CPECC codes that attain this upper bound by some combinatorial configurations (see Theorems \ref{thm:CPECC_construct_1},  \ref{thm:CPECC_construct_3}, Corollary \ref{cor:coj_solve_2}, and Lemmas \ref{lem:C'(n,9,4,2)}, \ref{lem:CPECC_construct_2}). Applying the asymptotic existence theorem of $(n,k,1)$-RBIBD, we also prove that when $n$ is sufficiently large, this upper bound is tight, which resolves the conjecture in \cite{lyj} (see Corollary \ref{cor:coj_solve_2}).
In Section 4, we firstly present an upper bound of $(n,t,w,w-2)$-LPECC codes (see Theorem \ref{thm:lpecc-upper}) by estimating the transversal number of graph using probabilistic method (see Theorem \ref{thm:clique-transversal-nonuniform}). We also establish the relationship between optimal $(n,t,w,w-2)$-LPECC codes and optimal $(n+1,t,w,w-2)$-CPECC codes (see Theorem \ref{lem:relation}).
In Section 5, we conclude this paper.


\section{Preliminaries}

Given a positive integer $n$, the set $\{1,2,\ldots,n\}$ is abbreviated as $[n]$.
The Hamming weight of a vector  $\boldsymbol v=(v_1,v_2,\ldots,v_n) \in \mathbb{F}^{n}_2$, denoted by
wt$(\boldsymbol v )$, is the number of nonzero positions in $\boldsymbol v $, while the support of $\boldsymbol v $ is defined as $supp(\boldsymbol v)\triangleq \{i\in [n]:v_i\neq0\}$.
The Hamming distance of $\boldsymbol v=(v_1,v_2,\ldots,v_n) \in \mathbb{F}^{n}_2$ and $\boldsymbol u=(u_1,u_2,\ldots,u_n) \in \mathbb{F}^{n}_2$, denoted by $d(\boldsymbol v,\boldsymbol u)$, is the number of positions where $\boldsymbol v$ and $\boldsymbol u$ differ.

\begin{definition}
For positive integers $n,t,w,e$ with $w>e$ and $n\geq t+w$, let ${\cal C}_1, {\cal C}_2,\ldots, {\cal C}_M$ are disjoint subsets of $\mathbb F_2^n$.
A set $\mathcal{C}={\cal C}_1\cup {\cal C}_2\cup\dots\cup{\cal C}_M$ is called an $(n,t,w,e)$\textbf{-LPECC code} of size $M$,
if $\mathcal{C}$ satisfies the following properties:
\begin{itemize}
\item[\textbf{A}$(t)$:] for every \(i\in[M]\) and every \(T\in\binom{[n]}{t}\), there exists \(\boldsymbol v\in C_i\) such that \(\operatorname{supp}(\boldsymbol v)\cap T=\varnothing\);

\item[\textbf{B}$(w)$:] {\em wt}$(\boldsymbol v )\leq w$ for any $\boldsymbol v\in \mathcal C$;

\item[\textbf{C}$(e)$:] $d(\boldsymbol v,\boldsymbol u)\geq2e+1$ for any $1\leq i<j\leq M$, $\boldsymbol v\in\mathcal C_i$ and $\boldsymbol u\in\mathcal C_j$.
\end{itemize}
An $(n,t,w,e)$\textbf{-CPECC code} is an $(n,t,w,e)$-LPECC code $\mathcal{C}={\cal C}_1\cup {\cal C}_2\cup\dots\cup{\cal C}_M$ if {\em wt}$(\boldsymbol v )=w$ for any $\boldsymbol v\in \mathcal C$.
\end{definition}

As a binary vector $\boldsymbol v\in \mathbb{F}^{n}_2$ naturally corresponds to a subset of an $n$-element set, we give the following combinatorial characterizations of LPECC codes and CPECC codes. For an $(n,t,w,e)$-LPECC code $\mathcal{C}={\cal C}_1\cup {\cal C}_2\cup\ldots\cup {\cal C}_M$, let ${\cal P}_i=\{supp(\boldsymbol{v})\colon \boldsymbol{v}\in {\cal C}_i\}$ for each $i\in [M]$ and ${\cal B}=\{supp(\boldsymbol{v})\colon \boldsymbol{v}\in {\cal C}\}$. Note that ${\cal P}_1,{\cal P}_2,\dots,{\cal P}_M$ are a partition of ${\cal B}$.
Such a set ${\cal B}$ with a partition of $M$ parts ${\cal P}_i$ is a set-theoretic representation of an $(n,t,w,e)$-LPECC code.
Formally, an $(n,t,w,e)$-LPECC code of size $M$ is a
set ${\cal B}$ of subsets of $n$-element set $[n]$ with a partition of ${\cal B}$ into ${\cal P}_1,\ldots, {\cal P}_M$ if the following properties hold:
\begin{itemize}
	\item[\textbf(1)] (Property $\textbf{A}(t)$)~~for any $ T\subset [n]$ with $|T|=t$, there exists $B\in \mathcal P_i$ satisfying $B\cap  T=\varnothing$, for all $i\in\{1,2, \ldots,M\}$;
	
	\item[\textbf(2)] (Property $\textbf{B}(w)$)~~$|B|\leq w$ for any $B\in {\cal B}$;

	\item[\textbf(3)] (Property $\textbf{C}(e)$)~~$|B_i\Delta B_j|\geq 2e+1$, for any $1\leq i<j\leq M$, $B_i\in \mathcal P_i$ and $B_j\in \mathcal P_j$, where $B_i\Delta B_j=(B_i\cup B_j)\setminus (B_i\cap B_j)$.
\end{itemize}
If Property $\textbf{B}(w)$ is replaced with Property $\textbf{B'}(w)$: $|B|=w$ for any $B\in {\cal B}$, then an $(n,t,w,e)$-LPECC code is called an $(n,t,w,e)$-CPECC code.
Denote by $C(n,t,w,e)$ (resp. $C'(n,t,w,e)$)  the maximum size in any $(n,t,w,e)$-LPECC code  (resp. $(n,t,w,e)$-CPECC code).

Next, we introduce some combinatorial configurations that will be used in Sections \ref{sec:optimal_CPECC} and \ref{sec:optimal_LPECC}.
An {\em $(n,k,1)$-balanced incomplete block design} (BIBD) is a pair $(V, \mathcal{B})$, where $V$ is an $n$-set of points and $\mathcal{B}$ is a collection of $k$-subsets of $V$ (called blocks) such that each pair of points from $V$ occurs  in exactly one block.
\begin{theorem}[\cite{ha}]\label{thm:(n,5,1)-BIBD}
An $(n,5,1)$-BIBD exists if and only if $n\equiv 1,5\pmod {20}$.
\end{theorem}

The following theorem provides the asymptotic existence result of $(n,k,1)$-BIBDs.

\begin{theorem}[\cite{chang_BIBD}]\label{thm:asy_BIBD}
If $n\geq e^{k^{k^{2}}}$, $n(n-1)\equiv 0\pmod {k(k-1)}$ and $n\equiv 1\pmod {k-1}$, then there exists an $(n,k,1)$-BIBD.
\end{theorem}

An $(n,k,1)$-BIBD $(V, \mathcal{B})$ is called {\em resolvable} if its block set $\mathcal{B}$ admits a partition into parallel classes, each being a partition of the point set $V$.

\begin{theorem}[\cite{cd}]\label{thm:rbibd-5}
	An $(n,5,1)$-RBIBD exists for $n\equiv  5 \pmod {20}$, with possible exceptions $n\in \{45, 345, 465, 645\}$.
\end{theorem}

\begin{theorem}[\cite{gr}]\label{thm:rbibd-7}
	An $(n,7,1)$-RBIBD exists for $n\equiv  7 \pmod {42}$ and $n\geq 294469$.
\end{theorem}

The following theorem shows the asymptotic existence result of a resolvable balanced incomplete block design.

\begin{theorem}[\cite{chang}]\label{thm:asy_RBIBD}
There exists an $(n,k,1)$-RBIBD for any integer $n>e^{e^{k^{12k^2}}}$ with $n\equiv 0\pmod {k}$ and $n\equiv 1\pmod {k-1}$.
\end{theorem}

Given two graphs $G$ and $H$, a graph family $\mathcal F=\{F_1,F_2,\dots, F_b\}$ is called an {\em $H$-packing of $G$}, if each $F_i$ is a copy of $H$ in $G$ for $i\in[b]$ and $F_1,F_2,\dots, F_b$ are pairwise edge-disjoint. The $H$-\textit{packing number} of $G$, denoted by $P(H,G)$, is the maximum cardinality of an $H$-packing of $G$. An {\em $(n, k, 1)$-packing} is a $K_k$-packing of $K_n$ and each copy of $K_k$ in an $(n, k, 1)$-packing is called a {\em block}. The following lemma shows the first Johnson bound of an $(n, k, 1)$-packing.

\begin{lemma}[\cite{cd}]\label{lem:J_bound}
Each $(n,k,1)$-packing with $b$ blocks satisfies that $b\leq \big\lfloor\frac{n}{k}\big\lfloor\frac{n-1}{k-1}\big\rfloor\big\rfloor$.
\end{lemma}

For a graph $G$, let $d_G(v)$ be the {\em degree} of the vertex $v$ in $G$, $\gcd(G)=\gcd\{d_G(v):v\in V(G)\}$ the greatest common divisor of degrees of the vertices in $G$, $v(G)$  and $e(G)$ the number of vertices and edges in $G$ respectively. The following lemma provides a result of an $H$-packing number of $K_n$ $P(H,K_n)$.

\begin{theorem}[\cite{BKLO16,CY97,GKO20}]\label{thm:graph-packing}
    Let $H$ be a finite simple graph without isolated vertices. For sufficiently large $n$, it holds that
    \[
    P(H,K_n)=\left\lfloor \frac{n\gcd(H)}{2e(H)} \left\lfloor \frac{n-1}{\gcd(H)} \right\rfloor \right\rfloor,
    \]
    unless $\gcd(H)\mid n-1$ and $n(n-1)/\gcd(H) \equiv b \pmod{2e(H)/\gcd(H)}$ for some $1\le b\le \gcd(H)$, in which case
    \[
    P(H,K_n)=\left\lfloor \frac{n\gcd(H)}{2e(H)} \left\lfloor \frac{n-1}{\gcd(H)} \right\rfloor \right\rfloor-1.
    \]
\end{theorem}

We conclude this section with the following simple yet useful inequality.
\begin{lemma}\label{lem:inequality_1}
Let $\ell$ and $w$ be integers with $\ell\geq0$ and $w\geq 2$. Then
\begin{align*}
    \frac{\ell-w+2}{\ell+1} \leq \frac{\ell}{2(2w-3)}.
\end{align*}
\end{lemma}

\begin{proof}
For every integer $\ell$, exactly one of the inequalities $\ell\geq 2w-3$ and $\ell \leq 2w-4$ holds, and so
\begin{align*}
    0 \leq (\ell-(2w-4))(\ell-(2w-3)) = \ell(\ell+1)-2(2w-3)(\ell-w+2).
\end{align*}
Since $\ell +1>0$ and $2w-3>0$,
$\frac{\ell-w+2}{\ell+1} \leq \frac{\ell}{2(2w-3)}.$\qed
\end{proof}

\section{Optimal $(n,t,w,w-2)$-CPECC codes}\label{sec:optimal_CPECC}
Zhao and Zhang \cite{z-z} showed that if $2\binom{w}{w-e}\geq \binom{w-t-1}{w-e}+\binom{w+t}{w-e}$ and $\binom{w+t-i}{w-e-i}|\binom{n-i}{w-e-i}$ for $0\leq i\leq w-e-1$, then
 $\lim\limits_{n\rightarrow\infty}\frac{C'(n,t,w,e)}{\binom{n}{w-e}/\binom{w+t}{w-e}}=1$. In this section, we first establish a better upper bound for $(n,1,w,w-2)$-CPECC codes. We also resolve the conjecture in \cite{lyj} using graph-theoretic methods.

\subsection{Optimal $(n,1,w,w-2)$-CPECC codes}\label{subsec:(n,1,w,w-2)-CPECC}

In this subsection, we investigate the optimal $(n,1,w,w-2)$-CPECC codes.
This subsection provides a new upper bound of $(n,1,w,w-2)$-CPECC codes and constructs a family of $(n,1,w,w-2)$-CPECC codes attaining the new upper bound for $w\geq4$.
Since the existence of optimal $(n,1,3,1)$-CPECC codes has been completely solved in \cite{lj}, we only consider the case $w\geq 4$ in this subsection.

Let $\mathcal B=\mathcal P_1\cup\dots\cup\mathcal P_M$ be the set-theoretic representative of an $(n,1,w,w-2)$-CPECC code, and let $G_i=(V_i,E_i)$ be a graph for every $i\in[M]$, where $V_i=\bigcup_{B\in\mathcal P_i}B$ and $E_i=\bigcup_{B\in\mathcal P_i}\{\{u,v\}:u,v\in B,u\neq v\}$.
It follows from Property $\textbf{C}(e)$ that for any $1\leq i<j\leq M$ and $B_i\in \mathcal P_i$, $B_j\in \mathcal P_j$,
\begin{align*}
	|B_i\cap B_j| = \frac{|B_i|+|B_j|-|B_i \Delta B_j|}{2} \leq \frac{2w-(2(w-2)+1)}{2}=\frac{3}{2},
\end{align*}
which implies that any two $w$-subsets from different $\mathcal{P}_i$'s have at most one common point.
So $E_1,E_2,\dots,E_M$ are pairwise disjoint. Recall that $e(G_i)=|E_i|$ for every $i\in[M]$.
Therefore $M\leq\frac{\binom{n}{2}}{\min_{i\in[M]}e(G_i)}$. Thus, to obtain an upper bound of $(n,1,w,w-2)$-CPECC codes, it suffices to determine $\min_{i\in[M]}e(G_i)$.

The following lower bound of $\min_{i\in[M]}e(G_i)$ is obtained by Zhao and Zhang (see \cite[Claim 5]{z-z}).
\begin{lemma}[\cite{z-z}]\label{lem:e(P_i)}
For an $(n,t,w,e)$-CPECC code $\mathcal B=\mathcal P_1\cup\dots\cup\mathcal P_M$, it holds that $e(G_i)\geq{w+t\choose w-e}.$
\end{lemma}

Lemma \ref{lem:stability} gives a characterization of $G_i$ attaining the lower bound in Lemma \ref{lem:e(P_i)} and a stability result. To prove Lemma \ref{lem:stability}, we need the following lemma.

\begin{lemma}\label{lem:key-lemma}
		Let $G$ be a graph with clique number $\kappa(G)=w$ and let $\mathcal{F}$ be a non-empty family of the largest cliques of $G$, where clique number of G is the size of the largest cliques of $G$. Then
		$$\Big| \bigcap_{F\in\mathcal{F}} F \Big| + \Big| \bigcup_{F\in\mathcal{F}} F \Big| \geq 2w.$$
	\end{lemma}
	\begin{proof}
		We prove the lemma by induction on $|\mathcal{F}|$. If $|\mathcal{F}|=1$, let $\mathcal{F}=\{F_0\}$. Then $|F_0|=w$ and $F_0$ is the largest clique of $G$. Consequently,	

$$\Big| \bigcap_{F\in\mathcal{F}} F \Big| + \Big| \bigcup_{F\in\mathcal{F}} F \Big| = |F_0| + |F_0| = 2w.$$
		Now suppose that the lemma holds for $|\mathcal{F}|=r-1\geq1$. Set
			$$A = \bigcap_{F\in\mathcal{F}}F,~~ B = \bigcup_{F\in\mathcal{F}}F.$$
		By the induction hypothesis, $|A|+|B|\geq 2w$.
		Set $\mathcal{F}'=\mathcal{F}\cup \{F'\}$, where $F'$ is a largest clique with $F'\notin \mathcal{F}$. Then $|\mathcal{F}'|=r$ and
	
			$$\bigcap_{F\in\mathcal{F}'}F = \Big(\bigcap_{F\in\mathcal{F}}F\Big)\cap F' = A\cap F',~~
			\bigcup_{F\in\mathcal{F}'}F = \Big(\bigcup_{F\in\mathcal{F}}F\Big)\cup F' = B\cup F'.$$

		Set $X=A\setminus F'$ and $Y=B\cap F'$. Since $X\cap F'=\varnothing$ and $Y\subseteq F'$, $X\cap Y=\varnothing$. Furthermore, $X\cup Y$ is a clique of $G$ according to the following facts:
		\begin{itemize}
			\item[(1)] $X\subseteq A$ and $A\subseteq F$ for any $F\in \mathcal F$ imply that any two vertices of $X$ are adjacent;
			\item[(2)]  $Y\subseteq F'$ and $F'$ being a clique imply that any two vertices of $Y$ are adjacent;
			\item[(3)] for two vertices $x\in X$ and $y\in Y\subseteq B$, there exists a clique $F_1\in \mathcal F$ satisfying $y\in F_1$. Note that $x\in A\subseteq F$ for any $F\in \mathcal F$. Thus $\{x,y\}\subseteq F_1$.
		\end{itemize}
		It follows from $X\cap Y=\varnothing$, $\kappa(G)=w$ and $F'$ is a largest clique of $G$ that
		 $$|X|+|B\cap F'|= |X|+|Y| = |X\cup Y|  \leq w = |F'| = |F'\setminus B|+|B\cap F'|,$$
		and so $|F'\setminus B|\geq|X|=|A\setminus F'|$. It is readily to check that
			\begin{align*}
           \Big| \bigcap_{F\in\mathcal{F}'}F \Big| + \Big| \bigcup_{F\in\mathcal{F}'}F \Big|
			 =&|A\cap F'|+|B\cup F'|\\
			 =& (|A|-|A\setminus F'|)+(|B|+|F'\setminus B|)\\
			 =& |A|+|B|-|A\setminus F'|+|F'\setminus B|\\
		  \geq& |A|+|B| \geq 2w.
\end{align*}
This completes the proof. \qed
\end{proof}

A graph $G$ is called {\em $K_w$-free} if $G$ does not contain a subgraph isomorphic to $K_w$.

\begin{lemma}\label{lem:stability}
Let $w\geq4$, and $G$ be a simple graph with $\delta(G)\geq w-1$, where $\delta(G)$ denotes the minimum degree of $G$. If there exists at least a complete graph $K_w$ after removing any a vertex in $G$, then
$e(G)\geq {w+1\choose 2}$, and equality occurs if and only if $G$ is a complete graph $K_{w+1}$. Furthermore, if $G$ is $K_{w+1}$-free, then $e(G)\geq 2{w\choose 2}$.
\end{lemma}

\begin{proof}Since there exists at least a complete graph $K_w$ after removing any a vertex in $G$, $v(G)\geq w+1$. For $v(G)=w+1$, $G$ is a complete graph $K_{w+1}$ since there exists at least a complete graph $K_w$ after removing any a vertex in $G$, and so $e(G)={w+1\choose 2}$. If $v(G)\geq w+2$, then $e(G)\geq{w\choose 2}+2(w-2)+1>{w+1\choose 2}$ by $\delta(G)\geq w-1$ and $w\geq4$.
To sum up, $e(G)= {w+1\choose 2}$  if and only if $v(G)=w+1$ and $G$ is a complete graph $K_{w+1}$.

Next, we prove that if $G$ is $K_{w+1}$-free, then $e(G)\geq 2{w\choose 2}$.
It is easy to know that there exists at least one $K_w$ in $G$. Since $G$ is $K_{w+1}$-free, there is no $K_{w+1}$ in $G$, and so $\kappa(G)=w$. Let $\mathcal F$ be the family of all the largest cliques of $G$. Note that
\begin{align*}
		\bigcap_{F\in \mathcal F}F=\varnothing,
\end{align*}
 otherwise there exists a vertex $x\in V(G)$ such that $x$ is contained in all the largest cliques of $G$, and so $G-x$ does not contain a $K_w$ as a subgraph, a contradiction. By Lemma \ref{lem:key-lemma} with $\kappa(G)=w$,
  \begin{align*}
    	\Big| \bigcap_{F\in\mathcal{F}} F \Big|+\Big| \bigcup_{F\in\mathcal{F}} F \Big| = 0+\Big| \bigcup_{F\in\mathcal{F}} F \Big| \geq 2w.
   \end{align*}
   Therefore, $|V(G)|\geq \Big| \bigcup_{F\in\mathcal{F}} F \Big| \geq 2w$. Now by $\delta(G)\geq w-1$ and the handshaking lemma,
  \begin{align*}
    	2e(G) = \sum_{v\in V(G)} d(v) \geq |V(G)|\delta(G) = |V(G)|(w-1) \geq 2w(w-1).
   \end{align*}
Thus $e(G)\geq w(w-1)=2\binom{w}{2}$.
\qed
\end{proof}


\begin{theorem}\label{thm:ub_C'(n,1,w,w-2)}
	Let $n,w$ be integers with $w\geq4$. Then $$C'(n,1,w,w-2)\leq\left\lfloor\frac{n}{w+1}\left\lfloor\frac{n-1}{w}\right\rfloor\right\rfloor+\left\lfloor\frac{n(n-1)-\left\lfloor\frac{n}{w+1}\left\lfloor\frac{n-1}{w}\right\rfloor\right\rfloor(w+1)w}{2w(w-1)}\right\rfloor.$$
\end{theorem}
\begin{proof} Let $\mathcal B=\mathcal P_1\cup \mathcal P_2\cup\dots\cup\mathcal P_M$ be an $(n,1,w,w-2)$-CPECC code. Let $G_i=(V_i,E_i)$ be a graph for every $i\in[M]$, where $V_i=\bigcup_{B\in\mathcal P_i}B$ and $E_i=\bigcup_{B\in\mathcal P_i}\{\{u,v\}:u,v\in B,u\neq v\}$. By Property $\textbf{A}(t)$ and Lemma \ref{lem:stability}, each vertex of $V_i$ has degree at least $w-1$ and $e(G_i)\geq {w+1\choose 2}$.
Without loss of generality, suppose that $e(G_i)={w+1\choose 2}$ for all $i\in[b]$, $e(G_i)>{w+1\choose 2}$ and $G_i$ is $K_{w+1}$-free for all $b+1\leq i\leq M$.

By Lemma \ref{lem:stability} and Property $\textbf{C}(e)$, $([n],\{V_i:i\in[b]\})$ is an $(n,w+1,1)$-packing and all $E_i$ are pairwise disjoint.
It follows from Lemma \ref{lem:J_bound} that $$b\leq \big\lfloor\frac{n}{w+1}\big\lfloor\frac{n-1}{w}\big\rfloor\big\rfloor.$$
By Lemma \ref{lem:stability}, $e(G_i)\geq2{w\choose 2}$ for all $b+1\leq i\leq M$, and so
$$b\binom{w+1}{2}+2(M-b)\binom{w}{2}\leq\sum_{i=1}^Me(G_i)\leq\binom{n}{2}.$$
Thus $M=b+(M-b)\leq b+\frac{{n \choose 2}-{w+1\choose 2}b}{2{w\choose 2}}=b(1-\frac{{w+1\choose 2}}{2{w\choose 2}})+\frac{{n \choose 2}}{2{w\choose 2}}$. Since $1-\frac{{w+1\choose 2}}{2{w\choose 2}}>0$ by $w\geq4$, we have that
$$M\leq b\left(1-\frac{{w+1\choose 2}}{2{w\choose 2}}\right)+\frac{{n \choose 2}}{2{w\choose 2}}\leq  \big\lfloor\frac{n}{w+1}\big\lfloor\frac{n-1}{w}\big\rfloor\big\rfloor\left(1-\frac{{w+1\choose 2}}{2{w\choose 2}}\right)+\frac{{n \choose 2}}{2{w\choose 2}}.$$
Since $M$ is an integer, $M\leq\left\lfloor\frac{n}{w+1}\left\lfloor\frac{n-1}{w}\right\rfloor\right\rfloor+\left\lfloor\frac{n(n-1)-\left\lfloor\frac{n}{w+1}
\left\lfloor\frac{n-1}{w}\right\rfloor\right\rfloor(w+1)w}{2w(w-1)}\right\rfloor.$
\qed
\end{proof}

The following theorem presents a construction of $(n,1,w,w-2)$-CPECC codes attaining the upper bound in Theorem \ref{thm:ub_C'(n,1,w,w-2)}.
\begin{theorem}\label{thm:bibd_C'(n,1,w,w-2)}
For $w\geq4$, if there exists an $(n+1,w+1,1)$-BIBD, then $$C'(n,1,w,w-2)=\left\lfloor\frac{n}{w+1}\left\lfloor\frac{n-1}{w}\right\rfloor\right\rfloor+\left\lfloor\frac{n(n-1)-
\left\lfloor\frac{n}{w+1}\left\lfloor\frac{n-1}{w}\right\rfloor\right\rfloor(w+1)w}{2w(w-1)}\right\rfloor.$$
\end{theorem}
\begin{proof}
Suppose that $(X \cup \{\infty\}, \mathcal{F})$ is an $(n+1, w+1, 1)$-BIBD.
Let $\mathcal{F}_\infty = \{ F \in \mathcal{F} : \infty \in F \}$, $\mathcal{F}_1 = \{F \setminus \{\infty\} : F \in \mathcal{F}_\infty \}$ and $\mathcal{F}_2 = \mathcal{F} \setminus \mathcal{F}_\infty$. Note that $|\mathcal F|=\frac{(n+1)n}{(w+1)w}$, $|\mathcal{F}_\infty|=\frac{n}{w}$, and
\begin{align}\label{equa:num_block}
|\mathcal{F}_2|=|\mathcal F|-|\mathcal{F}_\infty|=\frac{(n-w)n}{(w+1)w}=\frac{n}{w+1}\left\lfloor\frac{n-1}{w}\right\rfloor
=\left\lfloor\frac{n}{w+1}\left\lfloor\frac{n-1}{w}\right\rfloor\right\rfloor.
\end{align}
By (\ref{equa:num_block}),
$$\left\lfloor\frac{n(n-1)-
\left\lfloor\frac{n}{w+1}\left\lfloor\frac{n-1}{w}\right\rfloor\right\rfloor(w+1)w}{2w(w-1)}\right\rfloor=\left\lfloor\frac{n}{2w}\right\rfloor.$$
Let $\mathcal F_1=\{F_1,\dots,F_{\frac{n}{w}}\}$. Note that $F_1,\dots,F_{\frac{n}{w}}$ are pairwise disjoint.
Let $$\mathcal P_F=\binom{F}{w}, \text{ for each } F\in\mathcal F_2,$$ and
$$\mathcal P_i=\{F_i,F_{i+\left\lfloor\frac{n}{2w}\right\rfloor}\}, \text{ for each }1\leq i\leq\left\lfloor\frac{n}{2w}\right\rfloor.$$

It is readily checked that $\mathcal B=(\bigcup_{F\in\mathcal F_2}\mathcal P_F)\cup(\bigcup_{1\leq i\leq\left\lfloor\frac{n}{2w}\right\rfloor}\mathcal P_i)$ is an $(n,1,w,w-2)$-CPECC code of size $\left\lfloor\frac{n}{w+1}\left\lfloor\frac{n-1}{w}\right\rfloor\right\rfloor+\left\lfloor\frac{n(n-1)-
\left\lfloor\frac{n}{w+1}\left\lfloor\frac{n-1}{w}\right\rfloor\right\rfloor(w+1)w}{2w(w-1)}\right\rfloor$.
\qed
\end{proof}

Combining Theorems \ref{thm:(n,5,1)-BIBD}, \ref{thm:asy_BIBD} and \ref{thm:bibd_C'(n,1,w,w-2)}, we obtain some new families of optimal $(n,1,w,w-2)$-CPECC codes.
\begin{corollary}\label{cor:C'(n,1,w,w-2)}
 Let $n$ and $w$ be positive integers.
\begin{enumerate}
\item[$(1)$] If $n\equiv 0,4\pmod {20}$, then $C'(n,1,4,2)=\left\lfloor\frac{n}{5}\left\lfloor\frac{n-1}{4}\right\rfloor\right\rfloor+\lfloor\frac{n}{8}\rfloor$.
\item[$(2)$] If $n\geq e^{{(w+1)}^{{(w+1)}^{2}}}-1$, $w\geq4$, $n(n+1)\equiv 0\pmod {w(w+1)}$ and $n\equiv 0\pmod {w}$, then $C'(n,1,w,w-2)=\left\lfloor\frac{n}{w+1}\left\lfloor\frac{n-1}{w}\right\rfloor\right\rfloor+\left\lfloor\frac{n}{2w}\right\rfloor $.
\end{enumerate}
\end{corollary}
\subsection{The upper bound of $(n,t,w,w-2)$-CPECC codes}
This subsection confirms the upper bound part of the following conjecture in \cite{lyj} by the $K_w$-transversal number and induced degenerate subgraphs of graph.

\begin{conjecture}[\cite{lyj}]\label{conj-1}
Let $n,t,w$ be positive integers. Then
 $$C'(n,t,w,w-2)\leq \frac{n(n-1)}{2(2w-3)(t+1)}.$$
Furthermore,
when $n>e^{e^{(2w-3)^{12(2w-3)^2}}}$, $\left \{
\begin {aligned}
&n\equiv 0\pmod {l(2w-3)}\\
&n\equiv 1\pmod {(2w-4)}\\
\end {aligned}
\right.$ and $t=l(w-2)-1$, it holds that
 $$C'(n,t,w,w-2)= \frac{n(n-1)}{2(2w-3)(t+1)}.$$
\end{conjecture}

Similar to subsection \ref{subsec:(n,1,w,w-2)-CPECC}, to give an upper bound of $(n,t,w,w-2)$-CPECC codes, it suffices to prove Lemma \ref{lem:CPECC_e(G_i)}. We begin with several definitions and preliminary results to prove Lemma \ref{lem:CPECC_e(G_i)}.
Let $G=(V,E)$ be a graph and $w\geq 1$. A subset $T\subseteq V$ is called a {\em $K_w$-transversal} of $G$ if $T$ intersects the vertex set of every copy of $K_w$ in $G$; that is, for every subgraph $Q\subseteq G$ with $Q\cong K_w$, $V(Q)\cap T\neq \varnothing$. The minimum cardinality among all $K_w$-transversals of $G$, denoted by $\tau_w(G)$, is called the {\em $K_w$-transversal number} of $G$. Equivalently,
$$\tau_w(G)=\min\bigl\{|T|: T\subseteq V,\; V(Q)\cap T\neq\varnothing \text{ for every } Q\subseteq G \text{ with } Q\cong K_w\bigr\}.$$
A graph $H$ is {\em $f$-degenerate} if every non-empty subgraph of $H$ contains a vertex of degree smaller than $f$.
Let $\alpha_{f}(G)$ denote the maximum number of vertices of an induced $f$-degenerate subgraph of $G$.

\begin{theorem}[\cite{Alon}]\label{thm:alon_induced}
	If $G=(V,E)$ is a graph, then
	\begin{align*}
		\alpha_{f}(G) \geq \sum_{v\in V} \min\Big\{1,\frac{f}{d_G(v)+1}\Big\}.
	\end{align*}
\end{theorem}
	
\begin{lemma}\label{lem:K_w-free}
	Let $H$ be a graph and $w\geq 3$. If $H$ is $(w-1)$-degenerate, then $H$ is $K_w$-free.
\end{lemma}

\begin{proof}
Suppose for contradiction that $H$ contains a copy of $K_w$. Let $Q\subseteq H$ be a non-empty subgraph with $Q\cong K_w$.
Note that $d_Q(x)=w-1$ for every $x\in V(Q)$.
Since $H$ is $(w-1)$-degenerate, every non-empty subgraph of $H$ must contain a vertex of degree strictly less than $w-1$, and so there must exist a vertex $x\in V(Q)$ satisfying that $d_Q(x)<w-1$. This is a contradiction.\qed
\end{proof}

\begin{lemma}\label{lem:CPECC_e(G_i)}
Let $G$ be a graph and $w\geq 3$ an integer. If there exists at least a complete graph $K_w$ after removing any $t$ vertices in $G$, then
$$e(G)\geq (2w-3)(t+1).$$
\end{lemma}
\begin{proof}
By the definition of $\alpha_{f}(G)$ and Theorem \ref{thm:alon_induced} with $f=w-1$, there exists an induced $(w-1)$-degenerate subgraph $H\subseteq G$ such that
\begin{align*}
    v(H) \geq \sum_{v\in V(G)} \min\Big\{1,\frac{w-1}{d_G(v)+1}\Big\}.
\end{align*}
Then
\begin{align*}
    |V(G)\setminus V(H)| \leq \sum_{v\in V(G)} \Big(1-\min\Big\{1,\frac{w-1}{d_G(v)+1}\Big\}\Big).
\end{align*}
For a vertex $v\in V(G)$, if $d_G(v) \leq w-2$, then
\begin{align*}
    1-\min\Big\{1,\frac{w-1}{d_G(v)+1}\Big\}=1-1=0;
\end{align*}
and if $d_G(v)\geq w-1$, then
\begin{align*}
    1-\min\Big\{1,\frac{w-1}{d_G(v)+1}\Big\}=1-\frac{w-1}{d_G(v)+1}=\frac{d_G(v)-w+2}{d_G(v)+1} \leq \frac{d_G(v)}{2(2w-3)}
\end{align*}
by Lemma \ref{lem:inequality_1} with $\ell=d_G(v)$.
So 
\begin{align*}
|V(G)\setminus V(H)| \leq \sum_{v\in V(G)} \Big(1-\min\Big\{1,\frac{w-1}{d_G(v)+1}\Big\}\Big)
    		         \leq \sum_{v\in V(G)} \frac{d_G(v)}{2(2w-3)}=\frac{e(G)}{2w-3}.
\end{align*}

Since $H$ is $(w-1)$-degenerate, $H$ is also $K_w$-free by Lemma \ref{lem:K_w-free}. Thus $T=V(G)\setminus V(H)$ is a $K_w$-transversal of $G$. (Otherwise, there exists a contradiction that $H$ contains at least a $K_w$). Therefore,
\begin{align*}
    \tau_w(G) \leq |T| = |V(G)\setminus V(H)| \leq \frac{e(G)}{2w-3}.
\end{align*}
Since there exists at least a complete graph $K_w$ after removing any $t$ vertices in $G$, $\tau_w(G)\geq t+1$. Thus $e(G)\geq (2w-3)(t+1)$.\qed
\end{proof}

We are ready to prove the upper bound part of the Conjecture \ref{conj-1} by employing Lemma \ref{lem:CPECC_e(G_i)}.
\begin{theorem}{\label{thm:ub_(n,t,w,w-2)-CPECC}}
For integers $n,t$ and $w$ with $w\geq3$, $C'(n,t,w,w-2)\leq \frac{n(n-1)}{2(2w-3)(t+1)}$.
\end{theorem}

\begin{proof}
Let $\mathcal{C}={\cal C}_1\cup {\cal C}_2\cup\ldots\cup {\cal C}_M$ be an $(n,t,w,w-2)$-CPECC code and ${\cal B}={\cal P}_1\cup\ldots\cup{\cal P}_M$ the corresponding set-theoretic representation of $\mathcal{C}$ over $[n]$.
Let $G_i=(V_i,E_i)$ be a graph for every $i\in[M]$, where $V_i=\bigcup_{B\in\mathcal P_i}B$ and $E_i=\bigcup_{B\in\mathcal P_i}\{\{u,v\}:u,v\in B,u\neq v\}$.
Since every two codewords from different $\mathcal{C}_i$'s have Hamming distance at least $2(w-2)+1$, any two $w$-subsets from different $\mathcal{P}_i$'s have at most one common point, and so $E_1,E_2,\dots,E_M$ are pairwise disjoint.
Thus $$M\leq\frac{\binom{n}{2}}{\min_{i\in[M]}e(G_i)}.$$
It follows from Lemma \ref{lem:CPECC_e(G_i)} that $e(G_i)\geq(2w-3)(t+1)$ for each $i\in[M]$. Thus
$$M\leq \frac{\binom{n}{2}}{(2w-3)(t+1)}=\frac{n(n-1)}{2(2w-3)(t+1)}.$$
\qed
\end{proof}

\subsection{The construction of optimal $(n,t,w,w-2)$-CPECC codes}

This subsection confirms the construction part of the Conjecture \ref{conj-1} and provides some new families of optimal $(n,t,w,w-2)$-CPECC codes.

\begin{theorem}\label{thm:CPECC_construct_1}
	If there exists an $(n,2w-3,1)$-RBIBD, $w\geq3$, $l(2w-3)|n$ and $t=l(w-2)-1$, then
\begin{align*}
 C'(n,t,w,w-2)= \frac{n(n-1)}{2(2w-3)(t+1)}.
\end{align*}
\end{theorem}
\begin{proof}
Let $(X,\mathcal{B})$ be an $(n,2w-3,1)$-RBIBD,  which admits a resolution of parallel classes $\mathcal{P}_i=\big\{B_{i,h}: 1\leq h\leq\frac{n}{2w-3}\big\}$, $1\leq i\leq \frac{n-1}{2w-4}$. Set
\begin{align*}
\mathcal C_{i,j}=\bigcup_{0\leq r\leq l-1}\binom{B_{i,jl-r}}{w},\ \text{where }  1\leq i\leq \frac{n-1}{2w-4}, 1\leq j\leq \frac{n}{l(2w-3)},
\end{align*}
and
$$\mathcal C=\bigcup_{i=1}^{\frac{n-1}{2w-4}}\bigcup_{j=1}^{ \frac{n}{l(2w-3)}}\mathcal C_{i,j}.$$
Note that
$\frac{n-1}{2w-4}\times\frac{n}{l(2w-3)}=\frac{n(n-1)}{2l(w-2)(2w-3)}=\frac{n(n-1)}{2(2w-3)(t+1)}.$
It is routine to check that $\mathcal C$ is an $(n,t,w,w-2)$-CPECC code of size $\frac{n(n-1)}{2(2w-3)(t+1)}$.\qed
\end{proof}

Applying Theorems \ref{thm:asy_RBIBD}, \ref{thm:ub_(n,t,w,w-2)-CPECC} and \ref{thm:CPECC_construct_1}, the following corollary solves the construction part of the Conjecture \ref{conj-1}.
\begin{corollary}\label{cor:coj_solve_2}
When $w\geq3$, $n>e^{e^{(2w-3)^{12(2w-3)^2}}}$, $\left \{
\begin {aligned}
&n\equiv 0\pmod {l(2w-3)}\\
&n\equiv 1\pmod {(2w-4)}\\
\end {aligned}
\right.$ and $t=l(w-2)-1$, it holds that
\begin{align*}
 C'(n,t,w,w-2)= \frac{n(n-1)}{2(2w-3)(t+1)}.
\end{align*}
\end{corollary}

The following two lemmas give a construction of an optimal $(n,t,4,2)$-CPECC code and an optimal $(n,t,5,3)$-CPECC code by resolvable balanced incomplete block designs.
\begin{lemma}\label{lem:C'(n,9,4,2)}
Let $t\equiv 1 \pmod {2}$, $n\equiv 5 \pmod {20}$, $n\equiv 0\pmod {\frac{5(t+1)}{2}}$ and $n\not\in\{45,345,465,645\}$. Then
$$C'(n,t,4,2)= \frac{n(n-1)}{10(t+1)}.$$
\end{lemma}

\begin{proof}
It follows from Theorem \ref{thm:rbibd-5} that there exists an $(n,5,1)$-RBIBD. Apply Theorems \ref{thm:ub_(n,t,w,w-2)-CPECC} and  \ref{thm:CPECC_construct_1} with $w=4$. \qed
\end{proof}

\begin{lemma}\label{lem:CPECC_construct_2}
Let $t\equiv 2 \pmod {3}$, $n\equiv 7 \pmod {42}$, $n\equiv 0\pmod {\frac{7(t+1)}{3}}$ and $n\geq294469$. Then
$$C'(n,t,5,3)= \frac{n(n-1)}{14(t+1)}.$$
\end{lemma}
\begin{proof}
By Theorem \ref{thm:rbibd-7}, an $(n,7,1)$-RBIBD exists. The result
follows from Theorems \ref{thm:ub_(n,t,w,w-2)-CPECC} and \ref{thm:CPECC_construct_1} with $w=5$.
\qed
\end{proof}

The following construction of $(n,t,w,w-2)$-CPECC codes is obtained by means of graph packings.
\begin{theorem}\label{thm:CPECC_construct_3}
    Let $t,w$ be fixed positive integers with $w\ge 3$. Suppose that $n$ is sufficiently large. Then
    \begin{equation}\label{equa:cpecc}
        C'(n,t,w,w-2)=\left\lfloor \frac{n(n-1)}{2(2w-3)(t+1)} \right\rfloor,
    \end{equation}
    whenever one of the following conditions holds:
    \begin{itemize}
        \item[$(i)$]  $t\in \{(w-1)x+(w-2)y-1: x,y\in\mathbb{Z}^+\}$;
        \item[$(ii)$]  $w-1\mid t+1$, $2w-3\mid n-1$, and $\frac{n(n-1)}{2w-3}\not\equiv 1,2,\dots,2w-3 \pmod{2(t+1)}$;
        \item[$(iii)$]  $w-2\mid t+1$, $2w-4\mid n-1$, and $\frac{n(n-1)}{2w-4}\not\equiv 1,2,\dots,2w-4 \pmod{\frac{(2w-3)(t+1)}{w-2}}$.
    \end{itemize}
    In particular, when $t\ge (w-1)(w-2)$, \eqref{equa:cpecc} holds for all sufficiently large $n$.
\end{theorem}

\begin{proof}
    The upper bound part follows from Theorem \ref{thm:ub_(n,t,w,w-2)-CPECC}. Now we construct optimal CPECC codes by using Theorem~\ref{thm:graph-packing}.

    $(i)$ If $t\in \{(w-1)x+(w-2)y-1: x,y\in\mathbb{Z}^+\}$, then write $t=(w-1)x+(w-2)y-1$, where $x,y\in\mathbb{Z}^+$.
    Let $H$ be the vertex-disjoint union of $x$ copies of $K_{2w-2}$ and $y$ copies of $K_{2w-3}$. Note that $\gcd(H)=1$ and
      \[
      e(H)=\binom{2w-2}{2}x+\binom{2w-3}{2}y=(2w-3)(t+1).
      \]
    Let $\{H_1,\dots,H_m\}$ be a largest collection of edge-disjoint copies of $H$ in $K_n$, where $m=P(H,K_n)$. By Theorem~\ref{thm:graph-packing}, for sufficiently large $n$,
    \[
    m=P(H,K_n)=\left\lfloor \frac{n(n-1)}{2e(H)} \right\rfloor=\left\lfloor \frac{n(n-1)}{2(2w-3)(t+1)} \right\rfloor.
    \]

    For each $i\in [m]$, define
    \[
    \mathcal P_i=\left\{B\in\binom{V(H_i)}{w}: H_i[B]\cong K_w\right\},
    \]
    where $H_i[B]$ is an induced subgraph of $H_i$.
    It is easy to verify that $\cup_{i\in [m]} \mathcal P_i$ forms an $(n,t,w,w-2)$-CPECC code of size $\left\lfloor\frac{n(n-1)}{2(2w-3)(t+1)}\right\rfloor$.
    Therefore, for all sufficiently large $n$, when (i) holds,
     \[
     C'(n,t,w,w-2)=\left\lfloor\frac{n(n-1)}{2(2w-3)(t+1)}\right\rfloor.
     \]

    If (ii) (resp. (iii)) holds, then write $(w-1)x=t+1$ (resp. $(w-2)x=t+1$), and so let $H$ be the vertex-disjoint union of $x$ copies of $K_{2w-2}$ (resp. $K_{2w-3}$), and note that $\gcd(H)=2w-3$ (resp. $2w-4$). Then the conclusion holds similarly as (i).

In particular, condition~(i) holds for all \(t\ge (w-1)(w-2)\). Indeed, for any two coprime positive integers \(a_1\) and \(a_2\), the Frobenius number \(g(a_1,a_2)\) is the largest integer that cannot be expressed as a nonnegative integer combination of \(a_1\) and \(a_2\). This is the classical Frobenius coin problem, and Sylvester showed that \(g(a_1,a_2)=a_1a_2-a_1-a_2\)(see, e.g., \cite{Alf05}). Since \(\gcd(w-1,w-2)=1\), whenever
\[
t\ge (w-1)+(w-2)-1+\bigl(g(w-1,w-2)+1\bigr)
   =(w-1)(w-2),
\]
there exist \(x,y\in\mathbb Z^+\) such that \(t=(w-1)x+(w-2)y-1. \)
Thus, condition~(i) is automatically satisfied whenever
\(t\ge (w-1)(w-2)\).\qed
\end{proof}

\section{Optimal $(n,t,w,w-2)$-LPECC codes}\label{sec:optimal_LPECC}
This section gives an upper bound of $(n,t,w,w-2)$-LPECC codes and a construction of optimal $(n,t,w,w-2)$-LPECC codes derived from optimal $(n+1,t,w,w-2)$-CPECC codes.

To prove the upper bound of $(n,t,w,w-2)$-LPECC codes, we show the following strengthening of Lemma \ref{lem:CPECC_e(G_i)}.
Let $\mathcal F$ be a family of subsets of $[n]$.
A set $T\subseteq[n]$ is a {\em transversal} of a family $\mathcal F$ if $|T \cap F|\neq\emptyset$ for every $F\in\mathcal F$.
The {\em transversal number} of $\mathcal F$ is the minimum size of all transversals of $\mathcal F$.

\begin{theorem}\label{thm:clique-transversal-nonuniform}
    Let $w\ge 2$ be an integer. Let $G=(V,E)$ be a graph. Then for any subset $V_0\subseteq V$, the transversal number of the family \[
    \mathcal{F}_w(G;V_0)=\left\{A\in\binom{V}{w}: G[A]\cong K_w\right\}\cup \left\{B\in \binom{V_0}{w-1}:G[B]\cong K_{w-1}\right\}
    \]
    is at most $\frac{e(G)+|V_0|}{2w-3}$.
\end{theorem}
\begin{proof}
Choose a uniformly random linear ordering \(\prec\) of the vertex set \(V\).
Define
\[
T_{\prec}^{w-1}=\{v\in V\setminus V_0:|N^-_{\prec}(v)|\ge w-1\},
\]
and
\[
T_{\prec}^{w-2}=\{v\in V_0:|N^-_{\prec}(v)|\ge w-2\},
\]
where $N_G(v)=\{u\in V:\{u,v\}\in E\}$ and $N^-_{\prec}(v)=\{u\in N_G(v): u\prec v\}.$

For every copy of \(K_w\) in \(G\) \(Q\), let \(v\) be the last vertex of \(Q\) in the ordering \(\prec\).
The other \(w-1\) vertices of \(Q\) are all neighbors of \(v\), and all of them appear before \(v\). So \(|N^-_{\prec}(v)|\ge w-1\), which shows that $v\in T_{\prec}^{w-1}\cup T_{\prec}^{w-2}$. Hence $T_{\prec}^{w-1}\cup T_{\prec}^{w-2}$ is a transversal of $\{A\in\binom{V}{w}: G[A]\cong K_w\}$.
For every copy of \(K_{w-1}\) in \(G[V_0]\) \(Q\), let \(v\) be the last vertex of \(Q\) in the ordering \(\prec\).
The other \(w-2\) vertices of \(Q\) are all neighbors of \(v\), and all of them appear before \(v\). So \(|N^-_{\prec}(v)|\ge w-2\) and $v\in T_{\prec}^{w-2}$. Hence $T_{\prec}^{w-2}$ is a transversal of $\{B\in \binom{V_0}{w-1}:G[B]\cong K_{w-1}\}$.
To sum up, $T_{\prec}^{w-1}\cup T_{\prec}^{w-2}$ is a transversal of $\mathcal{F}_w(G;V_0)$.

    Fix a vertex \(v\), and write \(d=d_G(v).\) In a uniform random ordering, for any $v\in V\setminus V_0$,
    \[
      \mathbb{P}(v\in T_{\prec}^{w-1}) =
   \mathbb{P}(|N^-_{\prec}(v)|\ge w-1).
    \]
    Note that if \(d\le w-2\), then the probability is zero, otherwise,
    \[
    \mathbb{P}(|N^-_{\prec}(v)|\ge w-1)=\frac{d+1-(w-1)}{d+1}=\frac{d-w+2}{d+1}.
    \]
    Since \(\frac{d-w+2}{d+1}\le \frac{d}{2(2w-3)}\) by Lemma \ref{lem:inequality_1},
    \[
     \mathbb{E}[|T_{\prec}^{w-1}|]
    =
    \sum_{v\in V\setminus V_0}\mathbb{P}(v\in T_{\prec}^{w-1})
    \le
    \sum_{v\in V\setminus V_0} \frac{d_G(v)}{2(2w-3)}.
    \]
    For any \(v\in V_0\), if \(d\le w-3\), then the probability is zero. Otherwise,
    \[
      \mathbb{P}(v\in T_{\prec}^{w-2}) =
   \mathbb{P}(|N^-_{\prec}(v)|\ge w-2) \le
   \frac{d+1-(w-2)}{d+1}=\frac{d-w+3}{d+1}.
    \]
    Since $\frac{d-w+3}{d+2}\le \frac{d+1}{2(2w-3)}$ by Lemma \ref{lem:inequality_1} and $d+1>0$, $\frac{d-w+3}{d+1}\le \frac{d+2}{2(2w-3)}$. Note that
    \[
    \mathbb{E}[|T_{\prec}^{w-2}|]
    =
    \sum_{v\in V_0}\mathbb{P}(v\in T_{\prec}^{w-2})
    \le
    \sum_{v\in V_0} \frac{d_G(v)+2}{2(2w-3)}
    = \frac{|V_0|}{2w-3} + \sum_{v\in V_0} \frac{d_G(v)}{2(2w-3)}.
    \]
    Therefore,
    \[
    \mathbb{E}[|T_{\prec}^{w-1}\cup T_{\prec}^{w-2}|]
    \leq\mathbb{E}[|T_{\prec}^{w-1}|]+ \mathbb{E}[|T_{\prec}^{w-2}|]
    \leq \frac{|V_0|}{2w-3} + \sum_{v\in V} \frac{d_G(v)}{2(2w-3)}
    = \frac{e(G)+|V_0|}{2w-3}.
    \]
    Thus there exists $T_{\prec}^{w-1}\cup T_{\prec}^{w-2}$ satisfying $|T_{\prec}^{w-1}\cup T_{\prec}^{w-2}|\leq\frac{e(G)+|V_0|}{2w-3}$ with positive probability.

    Recall that $T_{\prec}^{w-1}\cup T_{\prec}^{w-2}$ is a transversal of $\mathcal{F}_w(G;V_0)$. By the definition of the transversal number, the transversal number of $\mathcal{F}_w(G;V_0)$ is at most $\frac{e(G)+|V_0|}{2w-3}$.\qed
\end{proof}

\begin{theorem}\label{thm:lpecc-upper}
    For any positive integers $n,t,w$ with $n\ge (2w-3)(t+1)$ and $w\ge 3$, it holds that
    \[
    C(n,t,w,w-2)\le \frac{n(n+1)}{2(2w-3)(t+1)}.
    \]
\end{theorem}
\begin{proof}
 Let $\mathcal{B}=\mathcal  P_1\cup\mathcal P_2\cup\dots\cup\mathcal P_M$ be an $(n,t,w,w-2)$-LPECC code, and let $G_i=(V_i,E_i)$ be a graph for every $i\in[M]$, where $V_i=\bigcup_{B\in\mathcal P_i}B$ and $E_i=\bigcup_{B\in\mathcal P_i}\{\{u,v\}:u,v\in B,u\neq v\}$.
 For each $i\in [M]$, let $\mathcal P_i^{w-1}=\mathcal P_i\cap \binom{[n]}{w-1}$ and denote $U^{w-1}(\mathcal P_i)=\bigcup_{B\in \mathcal P_i^{w-1}} B$.
 Define \[
  I_{\ge w-1}=\left\{i\in [M]: \mathcal P_i \subseteq\binom{[n]}{w}\cup \binom{[n]}{w-1}\right\}.
 \]
 By Property $\textbf{A}(t)$, for any $t$-set $T\subseteq[n]$ and any $i\in I_{\ge w-1}$, there exists some block $B\in \mathcal{P}_i$ such that $B\cap T =\emptyset$. This implies that the transversal number of $\mathcal{F}_w(G_i;U^{w-1}(\mathcal P_i))$ is at least $t+1$.
 It then follows from Theorem~\ref{thm:clique-transversal-nonuniform} that for any $i\in I_{\ge w-1}$,
 \begin{equation}\label{equa:e+v}
    |e(G_i)|+|U^{w-1}(\mathcal P_i)| \ge (2w-3)(t+1).
 \end{equation}

  Let $Y=\bigcup_{i\in I_{\ge w-1}} \bigcup_{B\in \mathcal P_i} B$. By Property $\textbf{C}(e)$, $|B_i\triangle B_j|\geq 2(w-2)+1=2w-3$  for any $B_i\in \mathcal P_i, B_j\in\mathcal P_j$ with $i\neq j$ and $i,j\in [M]$. It follows that $E_i\cap E_j=\emptyset$ and $U^{w-1}(\mathcal P_i)\cap U^{w-1}(\mathcal P_j)=\emptyset$ for any distinct $i,j\in I_{\ge w-1}$. Hence
  \[
    \sum_{i\in I_{\ge w-1}} |e(G_i)|\le \binom{|Y|}{2}\ \text{  and  }\sum_{i\in I_{\ge w-1}} |U^{w-1}(\mathcal P_i)|\le |Y|.
    \]
  Consequently, $\sum_{i\in I_{\ge w-1}} (|e(G_i)|+|U^{w-1}(\mathcal P_i)|)\le \binom{|Y|}{2}+|Y|=\frac{|Y|(|Y|+1)}{2}.$
  By (\ref{equa:e+v}),
  \begin{equation}\label{equa:I_w-1}
  |I_{\ge w-1}|\le \frac{|Y|(|Y|+1)}{2(2w-3)(t+1)}.
  \end{equation}

   If $I_{\ge w-1}=[M]$ then the desired conclusion is obtained. Suppose that $[M]\setminus I_{\ge w-1}\neq \emptyset$. By Property $\textbf{C}(e)$, it is easy to see that $|[M]\setminus I_{\ge w-1}|=1$, so let $[M]\setminus I_{\ge w-1}=\{i_0\}$.

   If $\emptyset\in\mathcal P_{i_0}$, then $\mathcal P_i\subseteq \binom{[n]}{w}$ for any $i\in [M]\setminus\{i_0\}$ by Property $\textbf{C}(e)$. This means that $\bigcup_{i\in [M]\setminus\{i_0\}} \mathcal P_i$ forms an $(n,t,w,w-2)$-CPECC code. By Theorem~\ref{thm:ub_(n,t,w,w-2)-CPECC},
   \[
   M-1\le C'(n,t,w,w-2)\le \frac{n(n-1)}{2(2w-3)(t+1)},
   \]
   and hence $M\le \frac{n(n+1)}{2(2w-3)(t+1)}$ since $n\ge (2w-3)(t+1)$. On the other hand, if $\emptyset\notin\mathcal P_{i_0}$, then there exists some block $B_0\in \mathcal P_{i_0}$ with $1\le |B_0|\le w-2$, which implies that $|Y|\le |[n]\setminus B_0|\le n-1$ by Property $\mathbf{{C}}(e)$. Therefore, by (\ref{equa:I_w-1}) and $n\ge (2w-3)(t+1)$,
   \[
   M= 1+|I_{\ge w-1}|\le 1+\frac{|Y|(|Y|+1)}{2(2w-3)(t+1)} \le \frac{n(n+1)}{2(2w-3)(t+1)}.
   \]
   This completes the proof. \qed
\end{proof}

The following theorem provides a construction of optimal $(n,t,w,w-2)$-LPECC codes derived from optimal $(n+1,t,w,w-2)$-CPECC codes.
\begin{theorem}\label{lem:relation}
    For any positive integers $n,t,w$ with $n\ge w+t$ and $w\ge 3$,
    \[
    C'(n+1,t,w,w-2)\le C(n,t,w,w-2).
    \]
\end{theorem}
\begin{proof}
    Let $\mathcal{B}=\mathcal P_1\cup\mathcal P_2\cup\dots\cup\mathcal P_M\subseteq 2^{[n+1]}$ be an $(n+1,t,w,w-2)$-CPECC code with $M=C'(n+1,t,w,w-2)$. Now we construct an $(n,t,w,w-2)$-LPECC code as follows.
    For each $i\in [M]$, define
    \[
    B' = \left\{
    \begin{array}{cc}
    B\setminus\{n+1\}, & \text{if } n+1\in B\in\mathcal P_i,\\
    B, & \text{otherwise,}
    \end{array} \right.
     \text{ and }\ \   \mathcal P'_i=\{B':B\in\mathcal P_i \}.
    \]


    Claim that $\mathcal B'=\mathcal P'_1\cup\mathcal P'_2\cup\dots\cup\mathcal P'_M\subseteq 2^{[n]}$ is an $(n,t,w,w-2)$-LPECC code. Indeed, Properties $\textbf{A}(t)$ and $\textbf{B}(w)$ clearly hold by definition. It remains to verify Property $\textbf{C}(e)$. Let $B'_i\in \mathcal P'_i, B'_j\in\mathcal P'_j$, where $i\neq j$ and $i,j\in [M]$.
    Let $B_i$ and $B_j$ be the unmodified subsets corresponding to $B'_i$ and $B'_j$, respectively.
    If $n+1\in B_i\cap B_j$ or $n+1\not\in B_i\cup B_j$, then
    $$|B'_i\triangle B'_j|=|B_i\triangle B_j|\geq 2w-3.$$
    Suppose that $n+1\in B_i,n+1\not\in B_j$ without loss of generality. Since $|B_i|=|B_j|=w$, $|B_i\triangle B_j|$ is even, and so $|B_i\triangle B_j|\geq2w-2$. Then $$|B'_i\triangle B'_j|=|B_i\triangle B_j|-1\geq 2w-3.$$

    Consequently, $C(n,t,w,w-2)\ge M=C'(n+1,t,w,w-2)$.\qed
\end{proof}

Combining Theorems \ref{lem:relation}, \ref{thm:CPECC_construct_3} and Corollary \ref{cor:coj_solve_2}, the following two corollaries hold.
\begin{corollary}\label{cor-4}
If $w\geq3$, $n>e^{e^{(2w-3)^{12(2w-3)^2}}}-1$, $n\equiv -1\pmod {l(2w-3)}$, $n\equiv 0\pmod {(2w-4)}$, $t=l(w-2)-1$, and $n\ge (2w-3)(t+1)$, then
\begin{align*}
 C(n,t,w,w-2)= \frac{n(n+1)}{2(2w-3)(t+1)}.
\end{align*}
\end{corollary}
\begin{corollary}\label{cor-5}
Let $t,w$ be fixed positive integers with $w\ge 3$. Suppose that $n$ is  sufficiently large. Then
    \begin{equation}\label{equa:cpecc-1}
        C(n,t,w,w-2)=\left\lfloor \frac{n(n+1)}{2(2w-3)(t+1)} \right\rfloor,
    \end{equation}
    whenever one of the following conditions holds:
    \begin{itemize}
        \item[$(i)$]  $t\in \{(w-1)x+(w-2)y-1: x,y\in\mathbb{Z}^+\}$;
        \item[$(ii)$]  $w-1\mid t+1$, $2w-3\mid n$, and $\frac{n(n+1)}{2w-3}\not\equiv 1,2,\dots,2w-3 \pmod{2(t+1)}$;
        \item[$(iii)$]  $w-2\mid t+1$, $2w-4\mid n$, and $\frac{n(n+1)}{2w-4}\not\equiv 1,2,\dots,2w-4 \pmod{\frac{(2w-3)(t+1)}{w-2}}$.
    \end{itemize}
    In particular, when $t\ge (w-1)(w-2)$, \eqref{equa:cpecc-1} holds for all sufficiently large $n$.
\end{corollary}

Remark that Lemmas \ref{lem:C'(n,9,4,2)} and \ref{lem:CPECC_construct_2} imply the existence of several new families of optimal $(n,t,w,w-2)$-LPECC codes by Theorem \ref{lem:relation}.

\section{Concluding remarks}
In this paper, we show new upper bounds of $(n,1,w,w-2)$-CPECC codes, $(n,t,w,w-2)$-CPECC/LPECC codes, and their optimal families. However, it is still an open problem to complete the existence of optimal $(n,t,w,w-2)$-CPECC/LPECC codes.

\section*{Acknowledgements}

We are grateful to Professors Sihuang Hu, Lijun Ji, Zixiang Xu and Xiande Zhang for their valuable discussions and suggestions.


\end{document}